\documentclass[letterpaper, 10 pt, conference]{ieeeconf}
\IEEEoverridecommandlockouts	 
\usepackage{cite}
\usepackage{bm}
\usepackage{amsmath,amssymb,amsfonts}
\usepackage{algorithmic}
\usepackage{textcomp}
\usepackage{graphicx,color}
\usepackage{mathrsfs}
\usepackage{comment}
\usepackage{float}
\usepackage[ruled]{algorithm2e}
\usepackage{soul}
\usepackage{url}
\usepackage{color}
\usepackage{dsfont}
 \usepackage{bbm}
\usepackage{booktabs}
\usepackage{array}
\usepackage[table]{xcolor} 
\usepackage{yfonts}
\usepackage[normalem]{ulem}
\newcommand{\stkout}[1]{\ifmmode\text{\sout{\ensuremath{#1}}}\else\sout{#1}\fi}
\usepackage{arydshln}
\usepackage{multicol}
\usepackage{amsmath}

\newtheorem{theorem}{Theorem}[section]
\newtheorem{lemma}[theorem]{Lemma}

\newtheorem{proposition}[theorem]{Proposition}

\newtheorem{remark}{Remark}

\newcommand{\until}[1]{\{1,\dots,#1\}}
\newcommand{\untill}[1]{\{0,\dots,#1\}}

\usepackage{tabularx}
\usepackage{multirow}
\usepackage{makecell}

\newcommand\aamsout{\bgroup\markoverwith{\textcolor{violet}{\rule[0.5ex]{2pt}{1pt}}}\ULon}

\newcommand{\real}{\mathbb{R}}

\newcommand{\mc}{\mathcal}

\DeclareSymbolFont{bbold}{U}{bbold}{m}{n}
\DeclareSymbolFontAlphabet{\mathbbold}{bbold}

\newcommand{\Z}{\mathbb{Z}}
\newcommand{\N}{\mathbb{N}}

\newcommand\oprocendsymbol{\hbox{$\square$}}
\newcommand\oprocend{\relax\ifmmode\else\unskip\hfill\fi\oprocendsymbol}

\renewcommand{\theenumi}{(\roman{enumi})}

\newcommand*{\QEDA}{\hfill\ensuremath{\blacksquare}}%

\makeatletter
\let\NAT@parse\undefined
\makeatother

\usepackage[urlcolor=blue,linkcolor=blue, citecolor=blue]{hyperref}

\graphicspath{{img/}}

\title{\bf \LARGE Scalable Associative Memory for Robust Pattern Storage }

\title{\bf \LARGE Kuramoto-Based Memory Networks Feature Exponential
  Capacity}

\title{\bf \LARGE Oscillatory Associative Memory with Exponential
  Capacity}

\author{Taosha Guo$^{\dagger}$, Arie Ogranovich$^{\dagger}$, Arvind R. Venkatakrishnan, Madelyn R. Shapiro, \\Francesco Bullo and Fabio Pasqualetti 
\thanks{This material is based upon work supported in part by ARO award W911NF-24-1-0228.
T. Guo and F. Pasqualetti are with the
Department of Mechanical Engineering at the 
University of California, Riverside, Riverside, CA, 92521, USA. Emails: \href{mailto:tguo023@ucr.edu}{\{\texttt{tguo023}},
\href{mailto:fabiopas@ucr.edu}{\texttt{fabiopas\}@ucr.edu}},
A. Ogranovich, A. R. Venkatakrishnan, M. R. Shapiro and F. Bullo are with the
Department of Mechanical Engineering at the 
University of California, Santa Barbara, Santa Barbara, 93106, USA. 
\href{mailto:arieogranovich@ucsb.edu}
{\{\texttt{arieogranovich}},
\href{mailto:arvindragghav@ucsb.edu}
{\texttt{arvindragghav}},
\href{mailto:madelynshapiro@ucsb.edu}
{\texttt{madelynshapiro}},
\href{mailto:bullo@ucsb.edu}{\texttt{bullo\}@ucsb.edu}},
}
}

\begin{document}

\maketitle
\thispagestyle{empty}

\begingroup
\renewcommand\thefootnote{}
\footnotetext{$^\dagger$ These authors contributed equally.}
\endgroup
\begin{abstract}
  The slowing of Moore's law and the increasing energy demands of
  machine learning present critical challenges for both the hardware
  and machine learning communities, and drive the development of novel
  computing paradigms. Of particular interest is the challenge of
  incorporating memory efficiently into the learning process. Inspired by how human brains store and retrieve information, associative
  memory mechanisms provide a class of computational methods that can
  store and retrieve patterns in a robust, energy-efficient manner.
  Existing associative memory architectures, such as the celebrated
  Hopfield model and oscillatory associative memory networks, store
  patterns as stable equilibria of network dynamics. However, the capacity (i.e. the number of patterns that a network can memorize
  normalized by their number of nodes) of existing oscillatory models have been shown to decrease with
  the size of the network, making them impractical for large-scale,
  real-world applications.
  In this paper, we propose a novel associative memory architecture
  based on Kuramoto oscillators. We show that the capacity of our
  associative memory network increases exponentially with network
  size and features no spurious memories. In addition, we present
  algorithms and numerical experiments to support these theoretical
  findings, providing guidelines for the hardware implementation of
  the proposed associative memory networks.

%Therefore, it is crucial to restore

%normal system functions by stabilizing the appropriate cluster synchronization patterns. 

%%%%%%%%  however most of the exsistant works  have limited capacitys. or need consistant injection to force the phase configures to fixed values. 
%Most existing studies focus on designing controllers based on state measurements to achieve system stabilization. However, in many real-world scenarios, measuring system states in real time, such as neuronal activity in the brain, poses significant challenges, rendering the stabilization of such systems difficult. 
%%%% To over come this difficulty, in this paper, we proposed a network stureture and pattern storing machnism    

%To overcome this challenge, in this article, we employ an open-loop control strategy, 
%vibrational control, which does not require any state measurements. We establish some sufficient conditions under which vibrational inputs stabilize cluster synchronization. Further, we provide a tractable approach to design vibrational control. Finally, numerical experiments are conducted to demonstrate our theoretical findings.
\end{abstract}
\section{Introduction}{\label{sec:intro}}

%Why does associate memory is important
%The limitations of Moore’s law and the growing energy demands of current AI models pose critical challenges to both hardware and machine learning community, there is a ergent need for more energy efficient computing hardware with high-performance for the increasingly complex tasks.  Inspired by  how human brain store and retrive memories,  associative memory mechanisms, in particular, offer a class of computational methods that can store and retrieve patterns in a robust and energy efficient manner. For instance, In the celebrated Hopfield model of associative memory  \cite{JJHopfield:82}, neurons take discrete values ($\pm 1$) to store patterns. When a new pattern is presented, the network retrieves the stored pattern that matches it the most.  Of interest is then the capacity, the maximum number of patterns per neuron that the network can memorize.    The memory capacity of the classic Hopfield  network and the existing oscillatory associative memory network  is proportional to $1 / \log(N)$, where $N$ is the number of neurons in the network \cite{RJM-ECP-ERR-SSV:87}. It decreases as the size of the network increase, meaning it can not store and retrieve patterns correctly for a numerous number of patterns.  

  The limitations of Moore’s law and the growing energy demands of current AI models pose critical challenges to both the hardware and machine learning communities. There is an urgent need for more energy-efficient computing hardware that maintains high performance for increasingly complex tasks. Inspired by how the human brain stores and retrieves memories, associative memory mechanisms offer a class of computational methods that can store and retrieve patterns in a robust and energy-efficient manner. For instance, in the celebrated Hopfield model of associative memory \cite{JJHopfield:82}, neurons take discrete values ($\pm 1$) to store patterns. When a new pattern is presented, the network retrieves the stored pattern that best matches it. Of particular interest is the memory capacity.
  %that is, the maximum number of patterns per neuron that a network can memorize. \margin{I would remove this sentence, since it's contained in Remark 2: The memory capacity of the classic Hopfield network, as well as that of existing oscillatory associative memory networks, is proportional to $1 / \log(n)$, where $n$ represents the cardinality the network \cite{RJM-ECP-ERR-SSV:87}.} 
  Most of the associative memory networks in the existing literature suffer from poor memory capacity, meaning that they cannot store and retrieve patterns correctly when presented with a large number of patterns.

 %In this paper, we propose a new network model of  Kumamoto oscillators to store and retrieve binary patterns,  where the network consists of $N$ numbers of cells, each cell consisys of $n$ numbers of oscillators with the same  $\omega_0$. We theoretically prove that  any  pattern consists of $\{\pm 1, 0\}$ of length $N$ can to mapped to a phase configuration, which corresponds to a stable equilibrium of the proposed network, for any $n\geq 5$.  Therefore we can achieve the memory capability of   $\frac{3^N}{Nn-N+1}$. This implies the memory capacity increases as the size of the network increases.  

In this paper, we introduce a novel network model for associative memory based on Kuramoto oscillators \cite{Kuramoto1975}.
%, which we refer to as the \textit{\textbf{1D honeycomb Network}}. 
The network consists of a sequence of
cycle networks interconnected in a way that subsequent cycles share
only one vertex. We fully characterize the stable phase-locked configurations of the proposed network.
%that can serve as memory states for memory binary patterns. Most 
Importantly, we show that the proposed network has no spurious memories and achieves a capacity that scales exponentially with the network cardinality.

\noindent \textbf{Related work.} 
The classical Hopfiled memory network suffers from poor scalability and capacity \cite{RJM-ECP-ERR-SSV:87}.
Several studies \cite{AS:97}, \cite{LP-IG-DG:86}, \cite{TN-YCL-FC:04}, \cite{TN-FCH-YCL:04} have explored various associative memory network architectures to improve  capacity.
In particular, \cite{TN-FCH-YCL:04, TN-YCL-FC:04} introduce a Kuramoto oscillator–based network for associative memory, where coupling weights are calculated  from the Hebbian learning rule for a given set of patterns. In this associative memory model a pattern can be stored and retrieved flawlessly if its phase configuration forms a stable phase-locked state in the network.   In addition, a second Fourier term with strength \(\epsilon\) is added to the network dynamics to improve the memory capacity. Moreover, when the memorized patterns are mutually orthogonal, \cite{RWH-KK:15} and \cite{XZ-ZL-XX:20} demonstrate that the network can effectively distinguish these patterns from most others for a small \(\epsilon\). However, the assumption of orthogonal patterns is often unrealistic.  \cite{ZL-XZ-XX:22} introduce an algorithm to subgroup the standard patterns and apply an orthogonal lift to each subgroup.  Nevertheless, existing dynamical models of associative memory feature a capacity that decreases with the network cardinality, which limits their use in practice.

\noindent\textbf{Contributions.} 
The main contributions of this article are as follows. First, we formalize the problem of designing  an associative memory network using Kuramoto oscillators. Next, we propose a novel network structure which consists of a sequence of
cycles interconnected in a way that subsequent cycles share
only one vertex. Most importantly,  we give analytical expressions of the stable phase-locked configurations of the proposed network. By mapping binary patterns to these stable phase-locked configurations, we show that the memory capacity of our network 
grows exponentially with the network cardinality and features no spurious memory states. We also compare the memory capacity of the proposed network to that of alternative network structures, confirming the effectiveness and capacity of the proposed structure.

%and to the best of our knowledge, our network outperforms these alternatives, especially as the network size increases.

\noindent \textbf{Organization of the paper.} 
The rest of the paper is organized as follows. Section \ref{sec: network model} contains our problem formulation. Section \ref{sec: associative_kuramoto_exponential} presents main results of characterizing  stable phase-locked configurations of our associative memory model, as well as methods to store and retrieve binary patterns.  Section \ref{sec: numerical_section}  contains numerical comparisons between the proposed model and other  associative memory structures. Section \ref{conclusion} concludes the paper.

\section{Problem formulation and Background}\label{sec: network model}
Associative memory is a mechanism that allows the human brain to store
memories and later retrieve them in association with external
stimuli. Among the many existing models of associative memory,
Hopfield networks put forth the hypothesis that memories correspond to
dynamically stable attractors of a dynamical process, and that
retrieval of memories happens through the convergence of such
dynamical process to the attractor closest to the current
stimulus. Following a large body of literature and in an attempt to
maintain biological plausibility, in this work we uphold Hopfield's
view of memories as dynamical attractors, yet propose and analyze an
oscillator-based model to store and retrieve memories.

 To formalize our discussion, consider a set of $n$ oscillators with Kuramoto dynamics interconnected as described by the
weighted graph $\mc G = (\mc V, \mc E)$ with adjacency matrix $A =
[a_{ij}]$. The dynamics of oscillator $i \in \mc V$ read as
\begin{align}\label{eq: Kuramoto model}
  \dot{\theta}_i = \omega_i+ \sum_{j=1}^n a_{ij} \sin(\theta_j - \theta_i),
\end{align} 
%\aomargin{The graphs we work with in this paper are unweighted and the oscillators have identical natural frequencies. I'd like to write a second equation just below this which specifies the model dynamics in this case. Is this okay?}
where $\omega_i$ and $\theta_i$ denote the oscillator's natural
frequency and phase, respectively, and $n = |\mc V|$ is the total
number of oscillators. A phase-locked configuration, or
(frequency-synchronized) equilibrium of the network \eqref{eq:
  Kuramoto model} is a set of oscillators' phases
$\theta_1^\text{eq},\dots, \theta_n^\text{eq}$ that are a stable
solution to the dynamics \eqref{eq: Kuramoto model} and satisfy
$\dot \theta_i^\text{eq} = \dot \theta_j^\text{eq}$ at all times and
for all $i,j \in \until{N}$.  A large body of literature has studied
the number, characteristics, and stability properties of phase-locked
configurations of networks of Kuramoto oscillators.

Associative memory can be implemented using a dynamical process such
as the Kuramoto network \eqref{eq: Kuramoto model} by creating an
association map between the asymptotically stable phase-locked configurations of the model and the
memories. Retrieval of a memory is completed by initializing the model sufficiently close to the associated phase configuration and allowing the system to converge via the asymptotic stability of the phase configuration. Let $x^1, \dots, x^p$
be the binary patterns to be memorized, where
$x^j = [x_1^j \; \cdots \; x_l^j]$ and $x_i^j  \in \{0,1\}$ \footnote{In the Hopfield network, $x_i^j \in \{-1,1\}$ are binary states.} 
for $i \in \until{l}$ and $j \in \until{p}$. 
A straightforward map between the phase-locked configurations of \eqref{eq: Kuramoto model} and the memory patterns
$x^i$ is obtained by requiring the  phases to equal either
$0$ or $\pi$ and using the cosine function to map these values to $1$
and $-1$, respectively. This scheme has been used in prior works \cite{LP-IG-DG:86}, \cite{TN-YCL-FC:04}, \cite{TN-FCH-YCL:04}, but the models considered have poor scalability and
capacity, as measured by the number of memories that can be
successfully stored in the model normalized by the dimension of the
model. In this work we propose a different mapping between binary 
patterns and the phase-locked configurations of \eqref{eq: Kuramoto model} and
construct associate memory models based on this mapping. We study their memory capacity, namely
\begin{align}
  C = p / n ,
\end{align}
where $p$ is the number of binary patterns that the network can memorize and $n$ is the cardinality of the network, and we show our model's memory capacity grows exponentially with the cardinality of the network. This is in contrast to existing dynamical models of associative memory, where the capacity decreases with the network cardinality.

\begin{remark}{\emph{\bf (Capacity of existing dynamical models of associative memory)}} The capacity of a classic Hopfield network has been shown to scale as $1/(2 \log n)$ for large networks of size $n$ \cite{RJM-ECP-ERR-SSV:87}. Associative network models based on oscillators display similar trends, even when the dynamics are modified to include an additional, higher frequency, interaction term \cite{TN-YCL-FC:04}. Surprisingly, these models feature a capacity that decays as the network grows, which stands as a fundamental limitation of this dynamics-based approach to associative memory. In this paper we present an oscillator-based architecture that overcomes this limitation.
  \oprocend
\end{remark}

%For the classical Hopfield model, the memory capacity is known to be $1/(2 \log n)$ for large $n$ \cite{RJM-ECP-ERR-SSV:87}.  By adding a second Fourier term of strength $\epsilon$ to the dynamics \eqref{eq: Kuramoto model}, and the capacity of the proposed network is improved to $2 \epsilon ^2 / \log{n}$ \cite{TN-YCL-FC:04}.

%\textbf{Algebraic Graph Theory on the Torus.} In preparation for Section \ref{sec: associative_kuramoto_exponential} we introduce the basic theory of winding vectors for Kuramoto oscillators \cite{Jafarpour2019Flows}. Winding vectors will provide a simple description of the correspondence between discrete memory vectors and continuous stable equilibria, as well as assist in the proof of Theorem \ref{thm: uniqueness}.
%To compare angles $\theta_i$

\section{Associative Kuramoto networks with exponential capacity}
\label{sec: associative_kuramoto_exponential}
In this section we propose an oscillatory associative memory network
with Kuramoto dynamics and show that this network features a capacity
that grows exponentially with the network cardinality. We start by
defining the proposed network structure, which we call \emph{1D
  honeycomb network}. A 1D honeycomb network consists of a sequence of
cycle networks interconnected in a way that subsequent cycles share
only one vertex (see Fig. \ref{fig: network structure}). To be
specific, a 1D honeycomb network is defined by two integers, the
number of cycles, $m$, and the number of nodes in each cycle, $n_\text{c}$,
and is the described by the graph
$\mc G_m^{n_\text{c}} = (\mc V_m^{n_\text{c}} , \mc E_m^{n_\text{c}})$
where $\mc V_m^{n_\text{c}} = \until{m(n_\text{c} -1) + 1}$ and $\mc
E_m^{n_\text{c}} = \mc E_1 \cup \mc E_2$, with
\begin{equation}\label{eq: edges of cells}
  \begin{aligned}
    \mc E_1 &=
    \left\{
      \begin{array}{cc}
        (i,i+1), & 1 \leq i \leq m(n_c-1)
      \end{array}
    \right\},
    \\
    \mc E_2 & =
    \left\{
      \begin{array}{cc}
        (i, i+n_c-1), & \text{} i = p (n_\text{c} -1)+1
      \end{array}
    \right\},                                       
  \end{aligned}
\end{equation}
and $\text{} p = \untill{m-1}$.

\begin{comment}
\textcolor{red}{
$\mc E_m^{n_c} = \mc E_1 \cup \mc E_2 \cup  \mc E_3 \cup \mc E_4  $
\begin{equation}
  \begin{aligned}
    \mc E_1 &=
    \left\{
      \begin{array}{cc}
        (i,i+2), & i = (p-1)n_c-p+1
      \end{array}
    \right\},
    \\
    \mc E_2 & =
    \left\{
      \begin{array}{cc}
        (i, i+n_c), & i = (p-1)n_c-p+1
      \end{array}
    \right\},         
    \\
    \mc E_3 & =
    \left\{
      \begin{array}{cc}
        (i, i+1), & (p-1)n_c - p+3\leq i \leq pn_c-p
      \end{array}
    \right\}, 
    \\
    \mc E_4 & =
    \left\{
      \begin{array}{cc}
        (i, i+1), & i<n_c
      \end{array}
    \right\}, 
  \end{aligned}
\end{equation}
}
\end{comment}
%and $\text{} p = \{2, \cdots, m\}$. \textcolor{red}{We denote the $n_c$-cycles of vertices by $\mathcal{C}_p$ for $p \in \{0,...,m-1\}$, i.e $\mathcal{C}_1 = (1,2,...,n_c)$ and $\mathcal{C}_p = ((n_c-1)p, (n_c-1)p+2, (n_c-1)p+3, ..., (n_c-1)p+n_c)$ for $p > 1$.}
\begin{figure}[t]
  \centering
  \includegraphics[width=0.85\columnwidth,trim={0cm 0cm 0cm
    0cm},clip]{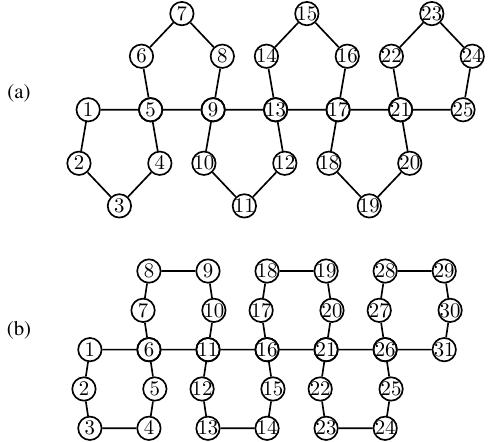}
  \caption{This figure shows the 1D honeycomb Network described by the
    graph $\mc H_m^{n_\text{c}}$. In (a), the network has $m=5$
    cycles, and each cycle is a pentagon with
    $n_c =5$ nodes. In (b), the network has $m=5$ cycles, and each cycle is a pentagon $n_c = 6$ nodes.}
  \label{fig: network structure}
\end{figure}

\begin{remark}{\emph{\bf (Alternative 1D honeycomb networks)}}{\label{remark: honeycomb networks}}
    The results presented in this paper hold also for networks different from the 1D honeycomb networks defined in \eqref{eq: edges of cells}. In fact, our results require only that adjacent cycles share only one node and no edges, as it is the case for the networks in Fig. \ref{fig: network
    structure}. However, the structure and numbering of the nodes defined in \eqref{eq: edges of cells} are convenient for our analysis and it do not restrict the implications of our results. \oprocend 
\end{remark}

\begin{figure}[t]
  \centering
  \includegraphics[width=0.95\columnwidth,trim={0cm 0cm 0cm
    0cm},clip]{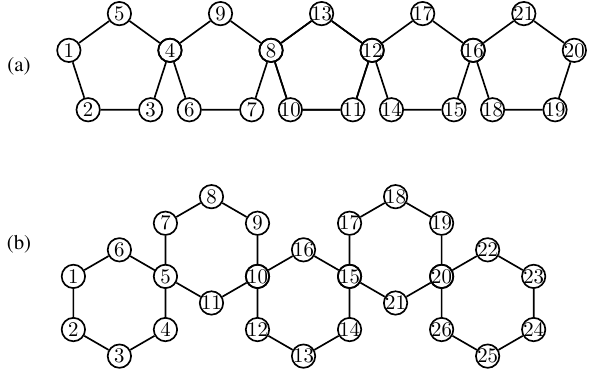}
  \caption{This figure demonstrates a different structure of the 1D
    honeycomb network described in \eqref{eq: edges of cells} In (a)
    we have $m=5$ and $n_c=5$. In (b) we have $m=5$ and $n_c=6$. }
  \label{fig: alternative network structure}
\end{figure}

To construct as associative memory network, we associate an oscillator
with each node of a 1D honeycomb network and assign the oscillators'
natural frequencies and interconnection weights. As we show in the
next result, the 1D honeycomb network model has at least exponentially many stable phase-locked configurations, and we find a simple characterization of many configurations; we will later see that there are no stable configurations aside from these.
% The reason we use the \emph{1D honeycomb Network} for associative
% memory is that we can easily design the weights \( a_{ij} \) for
% \( \mc G_m^{n_\text{c}} \) such that the network can achieve a large
% number of equilibria, which can be used to memorize patterns.  In the
% next theorem we gives the approach to design the network coupling
% weight $a_{ij}$, such that \( \mc G_m^{n_\text{c}} \) achieves a set

% of equilibria, and the number of that equilibria grows exponentially
% with the cardinality of the network.
\begin{theorem}{\emph{\bf (Phase-locked configurations of 1D honeycomb networks)}}{\label{thm: capacity}} 
Let $\mc G_m^{n_\text{c}}$ be a 1D honeycomb network with $m \ge 1$
  and $n_\text{c} \ge 5$, and let $A = [a_{ij}]$ be its adjacency
  matrix with $a_{ij} = c$, for any $c \in \real^{+}$, for all
  $(i,j) \in \mc E_{m}^{n_\text{c}}$. Consider the Kuramoto network
  \eqref{eq: Kuramoto model} with $\omega_i = \omega_j$ for all
  $i,j \in \mc V_{m}^{n_c}$. Then, the Kuramoto network has  $(2 \lceil \frac{n_c}{4} \rceil - 1)^{m}$ {locally exponentially stable} phase-locked configurations characterized by the
  following phase differences:%\margin{Not sure we need the intersection with [-pi/2,pi/2]. Maybe we can remove it from the statement}
  \begin{enumerate} {\label{condition}}

   \item[(C1)]\label{condition: thm 3.1 condition 1}
    $\theta_i - \theta_{j} = \frac{2\pi k}{n_c}$ with $k \in \Z$, $|k| \leq \lceil \frac{n_c}{4}\rceil - 1$, and
    $(i,j) \in \mc E_m^{n_\text{c}}$, and

  \item[(C2)] $\theta_i - \theta_{i+1} = \theta_j - \theta_{j+1}$ for
    all $p \in \{0,\dots,m-1\}$ and
    $i,j \in \{p(n_\text{c} -1) +1,\dots, (p+1)(n_\text{c} -1)\}$.
      \end{enumerate}
\begin{comment}
   \textcolor{red}{
    \item[(C2)] $\theta_i - \theta_{i+1} = \theta_j - \theta_{j+1}$ for
    all 
    $(i,i+1), (j,j+1) \in \mc E_3$}

       \textcolor{red}{
    \item[(C3)] $\theta_i - \theta_{i+1} = \theta_j - \theta_{j+1}$ for
    all 
    $(i,i+1), (j,j+1) \in \mc E_4$}

      \textcolor{red}{
    \item[(C4)] $\theta_i - \theta_{i+2} = \theta_{i+n_c} - \theta_{i} =  \theta_{i+2} - \theta_{i+3}  $ \\
    for
    all 
    $ i= (p-1)n_c-p+1$}
\end{comment}
 % which are characterized by the following: for any $n_c$-cycle $\mathcal{C}_p$ with winding number $k_p$ and whose ordered vertex sequence is denoted $(i_1, i_2, ..., i_{n_c})$, it holds that \textcolor{red}{
  %   \item[(CA1)] The angle differences along any edge are equal and given by $\theta_{i_l} - \theta_{i_{l+1}} = \frac{2\pi k_p}{n_c}$ \\
  %       where the indices $l, l+1$ are taken modulo $n_c$.}
   %  \item[(CA2)] $|k_p| \leq \lceil \frac{n_c}{4}\rceil - 1$
 %  \textcolor{red}{ where $p \in \{2, \cdots, m\}$}
%  Furthermore, all such equilibria are {\color{blue}exponentially stable}.\margin{Stability already mentioned in the body of Theorem. Maybe just replace "asymptotically" with "exponentially" and remove this last sentence}
  \end{theorem}
  \vspace{1em}
  \begin{proof}
  %\textcolor{blue}{
   %Step 1: Argue any configuratuions statisfy c1 and c2 is invariant in dynamics \eqref{eq: Kuramoto model}. 
  %frequency are the same. }
%By plugging  in the phase difference values characterized by (C1) and (C2) to the network dynamics given by \eqref{eq: Kuramoto model}, we obtain
%\( \dot{\theta}_i = \omega_i \)  for all $i \in \mc V_{m}^{n_c}$. Because we assume that $\omega_i =  \omega_j$ for all $i,j \in \mc V_{m}^{n_c}$, $\theta_i - \theta_j$ become invariant under the dynamics of \eqref{eq: Kuramoto model}. 
%Further, let $B \in \mathbb{R}^ {n \times |{\mc E}_{m}^{n_c}|}$, the Jacobian of the Kuramoto dynamics reads as \cite{Doerfler2014} $J = \frac{\partial}{\partial \bm  \theta} \dot{\theta} = -B\diag (a_{ij} \cos(\theta_i - \theta_j)) B^{\T}$. Because the phase differences  given by (C1) and (C2) are strictly less than $\frac{\pi}{2}$, the Jacobin at the any phase configurations given by (C1) and (C2) are locally exponentially stable \cite{Doerfler2014}.
Let the phases $\theta := (\theta_1,\dots,\theta_n)$ satisfy (C1) and (C2). Then, from \eqref{eq: Kuramoto model} and under the assumptions of the theorem, we have $\dot \theta_i = \dot \theta_j$ for all $i,j\in \mathcal{V}_{m}^{n_c}$, which implies that the system is in a phase-locked configuration. Next, let \(B \in \mathbb{R}^{n \times |\mathcal{E}_{m}^{n_c}|}\) be the incidence matrix of the network. According to~\cite{Doerfler2014}, the Jacobian of the Kuramoto dynamics is given by
\[
J({\theta})
= \frac{ \partial}{\partial {\theta}} \dot{\theta}
= -B \,\mathrm{diag}\!\bigl(a_{ij}\,\cos(\theta_i - \theta_j)\bigr)\,B^\top.
\]
\noindent Notice the phase differences described by (C1) and (C2) are strictly less than \(\pi/2\) due to the fact that 
\begin{align*}
|\theta_i - \theta_j| = \dfrac{2\pi |k|}{n_c}  &\leq \frac{2 \pi (\lceil \frac{n_c}{4} \rceil - 1)}{n_c}\\&< \dfrac{(2 \pi \frac{n_c}{4})}{n_c} = \pi/2.
\end{align*} Thus, all eigenvalues of the corresponding Jacobian are negative except for the zero-eigenvector in the direction of rotation \cite{Doerfler2014}, and the phase-locked configurations are locally exponentially stable.
\end{proof}
%For instance, if all phase differences are strictly smaller than π 2 (that is, the infinity-norm of x satisfies k x k1< π 2), then the Jacobian in Eq. (10) is known to be stable17.

  %\textcolor{blue}{
%Step 2: the Jacobin of (1) at the configurations at c1 and c2 is Lapalacian matrix with eigenvalues $<0$  when $n_c>=5$. 
 % stpe}

%The first condition on the equilibria in Theorem \ref{thm: capacity} \textcolor{red}{[finish defining phase-cohesive here]}

\begin{figure}[t]
  \centering
  \includegraphics[width=1\columnwidth,trim={0cm 0cm 0cm
    0cm},clip]{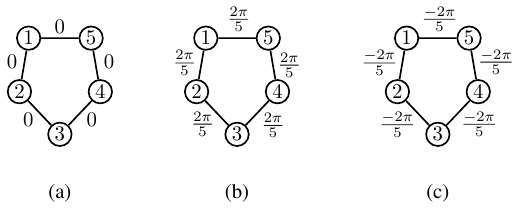}
  \caption{This figure shows the $3$  stable phase-locked configurations for $m=1$ and $n_c=5$. (a) The phases of the oscillators remains equal over time. (b) The phase differences  are locked to \(-\tfrac{2\pi}{5}\). (c) The phase differences are locked to \(\tfrac{2\pi}{5}\).}
  \label{fig: phase lock 0}
\end{figure}

Fig. \ref{fig: phase lock 0} shows an example of the stable phase-locked configurations with \(m=1\) and \(n_c=5\). 
 By mapping each binary pattern to a stable phase-locked configuration described in Theorem~\ref{thm: capacity}, the proposed 1D honeycomb Kuramoto network can represent up to
\begin{align}\label{eq: num equilibria}
    N_{\text{eq}} = \bigl(2\lceil \tfrac{n_c}{4} \rceil - 1\bigr)^m
\end{align}
distinct binary patterns. The memory capacity 
of the Kuramoto networks in \eqref{eq: edges of cells} is then given by
\begin{align}{\label{eq: capacity}}
C = \frac{p}{n}
= \frac{\bigl(2\lceil \tfrac{n_c}{4} \rceil - 1 \bigr)^m}{(n_c -1)m + 1},
\end{align}

where $p$ is the number of binary patterns  the network can memorize and $n = m(n_c-1)+1$ is the cardinality of the 1D honeycomb network. \eqref{eq: capacity} implies that the proposed network achieves exponential capacity with respect to the network cardinality $n$ because $p$ grows exponentially in $m$, while $n$ increases only linearly in $m$.

%To reach the network capacity, the binary vector should have length $l \geq  \lceil \log_{2}^{ (2\lceil \frac{n_c}{4}\rceil -1)^m} \rceil$.

To map a binary pattern to each stable phase-locked configuration, we can simply assign a unique index to each configuration and then convert that index into a binary pattern.
For example, when \(n_c = 5\) and \(m = 2\), Theorem~\ref{thm: capacity} implies that each cycle can independently have a phase difference of \(\tfrac{2\pi k}{n_c}\) for \(k \in \{-1,0,1\}\), yielding \(\bigl(2\lceil \tfrac{5}{4}\rceil - 1\bigr)^2 = 9\) distinct stable phase-locked configurations. The converting between the stable phase-locked configurations  and binary patterns is shown in Table~\ref{table: mapping}. Note that the mapping between binary patterns and choices of $k$ is not unique and can be chosen arbitrarily.

%Because \(k\) takes three possible values for each cycle, we can map 3-bit binary vectors to trit pairs, as shown in Table~\ref{table: mapping}.
% (pi - theta) 
% ()
%To fully use each equilibra, we use the $x^i = $ 

Theorem \ref{thm: capacity} describes a class of phase-locked configurations for which any angle difference along an edge is strictly less than $\pi/2$. Angle states $\theta$ which satisfy this property are referred to as \textit{phase-cohesive}, and can be classified via existing theory for Kuramoto models \cite{Jafarpour2019Flows}.

A key issue for associative memory networks is that the network might converge to spurious states other than the states that represent the desired memories. In the next theorem we show that the only stable phase-locked configurations of the Kuramoto network in Theorem \ref{thm: capacity} are those that are described by (C1) and (C2). The proof idea is to show that all stable phase-locked configurations of the 1D honeycomb topology are phase-cohesive, and then to use the theory of \cite{Jafarpour2019Flows} to classify them.
\begin{table}[t]
\label{table: mapping}
\centering
\begin{tabular}{l|l|c|c}
Phase differences & Scaled values $k$& Index & \text{ Binary patterns} \\
\hline
$ (-\frac{2\pi}{5},-\frac{2\pi}{5}) $ & $(-1, -1)$ & 1 & 0001 \\
$ (-\frac{2\pi}{5},0) $& $ (-1, 0)$  & 2 & 0010 \\
$ (-\frac{2\pi}{5},\frac{2\pi}{5}) $ & $ (-1, +1)$ & 3 & 0011 \\
$ (0,-\frac{2\pi}{5}) $ & $ (0, -1)$  & 4 & 0100 \\
$ (0,0) $ & $ (0, 0)$   & 5 & 0101 \\
$ (0,\frac{2\pi}{5}) $ & $ (0, +1)$  & 6 & 0110 \\
 $ (\frac{2\pi}{5},-\frac{2\pi}{5}) $ &$ (+1, -1)$ & 7 & 0111 \\
 $ (\frac{2\pi}{5},0) $ &$ (+1, 0)$  & 8 & 1000 \\
 $ (\frac{2\pi}{5},\frac{2\pi}{5}) $ &$ (+1, +1)$ & 9 & 1001 
\end{tabular}
\caption{Mapping from stable phase-locked configurations to binary patterns  and vice versa \cite{Dek:97}.}
\end{table}
%Using \eqref{eq: equilibra map}, Fig.\ \ref{fig: phase lock} illustrates how to map the stable phase-locked configurations from Theorem \ref{thm: capacity} to the patterns \(x = [0,1]\) and \(x = [-1,-1]\), respectively.
\begin{figure}[t]
  \centering
  \includegraphics[width=0.85\columnwidth,trim={0cm 0cm 0cm
    0cm},clip]{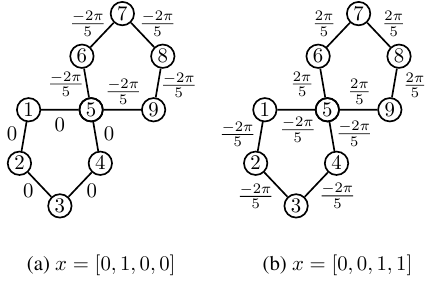}
  \caption{
  This figure shows the stable phase-locked configurations 
  of two different memory patterns for $m=2$ and $n_c=5$ (a). The phase-locked configuration that represents the binary pattern $x=[0,1,0,0]$.  (b). The phase-locked configuration  that represents the binary pattern $x = [0,0,1,1]$.
  }
  \label{fig: phase lock}
\end{figure}

\begin{theorem}{\emph{\bf (No spurious memories)}}{\label{thm: uniqueness}} 
    Consider the Kuramoto model \eqref{eq: Kuramoto model} with identical natural frequencies, and with a 1D honeycomb network. There are no stable phase-locked configurations of this system besides those already identified in Theorem \ref{thm: capacity}.
\end{theorem}

\begin{proof}
    As remarked in \eqref{eq: num equilibria}, there are $N_{\text{eq}} = \bigl(2\lceil \tfrac{n_c}{4} \rceil - 1 \bigr)^m$ stable phase-locked configurations identified in Theorem \ref{thm: capacity}. We will show that there are at most $N_{\text{eq}}$ such configurations, which would imply that there can be no other stable phase-locked configurations besides the $N_{\text{eq}}$ already found.

    Based on the assumption that every edge contains a node of degree 2, Proposition \ref{prop: phase cohesiveness of honeycombs} in the appendix shows that every stable phase-locked configuration of the model must have angle differences less than $\pi/2$ along any edge.
    
    The theory of winding vectors in \cite{Jafarpour2019Flows} allows us to bound the number of such configurations with angle difference less than $\pi/2$ along every edge. To every cycle $\sigma$ in a graph $\mc G$, one can associate a \textit{winding number} $w_{\sigma} \in \Z$. When the angle difference along any edge in $\sigma$ is strictly less than $\pi/2$, one can follow the proof of Theorem 3.3 in \cite{Jafarpour2019Flows} and show that
    \begin{align}\label{eq: winding number bound}
        -\biggr\lceil \frac{n_{\sigma}}{4} \biggr\rceil +1\leq w_{\sigma}  \leq \biggr\lceil \frac{n_{\sigma}}{4} \biggr\rceil -1,
    \end{align}
    where $n_{\sigma}$ denotes the number of edges in $\sigma$ (see \cite{Jafarpour2019Flows} for further explanation). When a sequence of cycles $(\sigma_1, \sigma_2, ..., \sigma_m)$ forms a cycle basis (see \cite{Jafarpour2019Flows}, \cite{royle2001algebraic} for a definition of cycle basis) we call $(w_{\sigma_1}, w_{\sigma_2}, ..., w_{\sigma_m})$ the \textit{winding vector} of $ \mc G$. For a 1D honeycomb network, the $m$ cycles of size $n_c$ form a cycle basis. It is shown in \cite{Jafarpour2019Flows} that all distinct phase-locked configurations with  have distinct winding vectors, so bounding the number of possible winding vectors will bound the number of phase-cohesive phase-locked configurations.
    
    Given the bound \eqref{eq: winding number bound} there are at most $2 \lceil \frac{n_{\sigma}}{4} \rceil - 1$ possible values of $w_{\sigma}$ at a phase-cohesive phase-locked configuration, hence there are at most
    $$\prod_{i = 1}^m \left(2\biggr\lceil \frac{n_{\sigma}}{4} \biggr\rceil -1 \right)$$
    phase-cohesive phase-locked configurations. In the case of the 1D honeycomb where the cycles in the cycle basis all have size $n_c$, the above bound simplifies to
    $$ \left(2\biggr\lceil \frac{n_{\sigma}}{4} \biggr\rceil -1 \right)^m = N_{\text{eq}}.$$
    Since Theorem \ref{thm: capacity} exhibited precisely $ N_{\text{eq}}$ stable phase-locked configurations there can be no other stable phase-locked configurations, as desired.
\end{proof}

\section{Numerical simulations of stable phase-locked configurations}
\label{sec: numerical_section}In this section we compare the capacity of the 1D honeycomb network \eqref{eq: edges of cells} to alternative network structures.

\begin{figure}
    \centering
\includegraphics[width=0.5\linewidth]{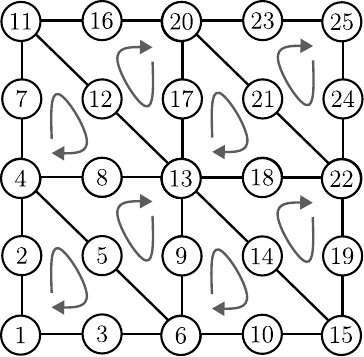}
    \caption{The 2D hexagonal array used in Section~\ref{sec: numerical_section}.}
    \label{fig: two_dimensional_hexagonal_graph}
\end{figure}

\begin{figure}[t]
    \centering
\includegraphics[width=\linewidth]{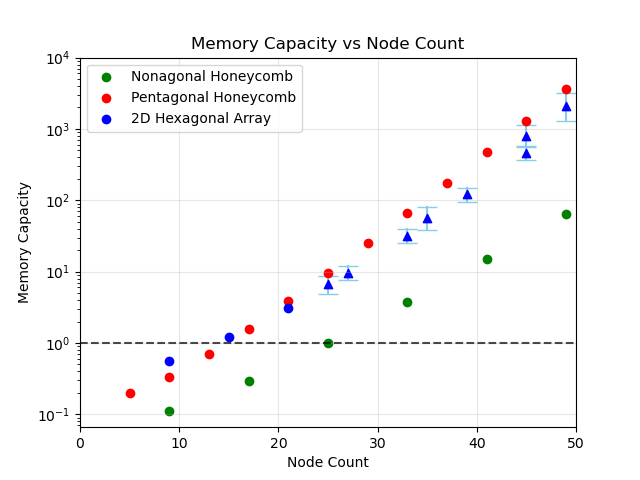}
    \caption{Comparison of the memory capacity of pentagonal and nonagonal honeycomb arrangements vs. the memory capacity of 2D hexagonal array. Triangular dots represent data points that were estimated via a sampling method, and the blue lines represent 95\% confidence intervals around each data point.}
    \label{fig:numerics}
\end{figure}

Using  dynamical systems for associative memory requires (i) precisely characterizing stable memory states that can represent patterns, and (ii) ensuring sufficient memory capacity to encode extensive information. To investigate whether other associative memory network structures might exhibit comparable memory capacity, we construct an alternative graph structure and compare its memory capacity with that of the 1D honeycomb structure. We test the memory capacity of both the pentagonal and nonagonal honeycomb networks, as they are the most node-efficient honeycomb structures.

In considering realistic models of associative memory that can be implemented on a single oscillatory chip layer, we consider planar graphs topologies \cite{ahmed2020oscillatorImplementation}. For scalability purposes, we are interested in graphs that have simple construction rules, such as the 1D honeycomb \eqref{eq: edges of cells} which is composed of copies of a single cycle. We thus compare the pentagonal and nonagonal honeycomb structures to a 2D hexagonal array array, as illustrated in Fig. \ref{fig: two_dimensional_hexagonal_graph}. We chose the hexagonal array as it is one of the smallest and most regular array-like networks; based on the theory in \cite{Jafarpour2019Flows} one can show that the triangular and square arrays have only $1$ phase-cohesive phase-locked configuration, so we would expect the hexagonal array to have more total phase-locked configurations than the triangular and square arrays.

%\tguomargin{Maybe we should say both are exponential but 1D honeycom grow faster}
%\aomargin{Thoughts on this version?}

In Fig. \ref{fig:numerics} we compute the memory capacity of each of the three oscillator networks, where the honeycomb structures we choose are described by Fig. \ref{fig: alternative network structure}. The stable phase-locked configurations are computed via the network flow solver in \cite{Jafarpour2019Flows}; this algorithm computes all stable, phase-cohesive phase-locked configurations by checking each possible winding vector. Since the graphs under consideration satisfy the hypotheses of Propositions \ref{prop: phase cohesiveness of our graphs} and \ref{prop: phase cohesiveness of honeycombs}, we know the network flow solver obtains all stable phase-locked configurations of the tested network models. Due to the exponentially growing runtime of looping over all winding vectors, we estimate the memory capacity of large 2D hexagonal arrays by sampling winding vectors and estimating the probability that a random winding vector is exhibited by a stable phase-locked configuration.

We observe that all three graphs structures have an exponential memory capacity in terms of  the network cardinality. The memory capacity of the pentagonal honeycomb is slightly larger than that of the 2D hexagonal array, although their growth rates are similar. The nonagonal honeycomb has a lower memory capacity for a comparable number of nodes, which can be seen by comparing the memory capacity in \eqref{eq: capacity} for $n_c=5,9$ when the number of cycles $m$ is chosen separately for each case such that the node-counts are equal. This demonstrates that the memory capacity of the pentagonal honeycomb structure is comparable with or larger than alternative, scalable network structures. A characterization of the largest memory capacity of any Kuramoto network given a fixed number of nodes remains under investigation. 

%It remains under investigation to characterize the largest memory capacity of a Kuramoto network with arbitrary topology but fixed number of nodes.

\section{Conclusions and future work}{\label{conclusion}}

In this paper, we investigate the problem of designing an oscillator-based network for associative memory, where each oscillator obeys simple Kuramoto dynamics. We show that our 1D honeycomb network model exhibits a capacity that scales exponentially with the number of oscillators and we completely characterize the  stable phase-locked configurations. We also show that such a network has no other stable phase-locked configurations that would lead to the retrieval of a spurious memory. We also derive simple formulas to map the phase-locked configurations of the network to binary memory patterns, which enables memory storage and retrieval. We provide numerical evidence that our associative memory network performs favorably against other regular oscillator-based associative memory networks. Directions of future research include a characterization of the basin of attraction of each phase-locked configuration, as well the design of associative memory networks with maximal capacity.

%that exhibits exponential memory capacity. Specifically, we propose the 1D honeycomb Kuramoto Network structure, which consists of \(m\) cycles and each cycle contains \(n_c\) oscillators. We provide analytical expressions for all the locally exponentially stable equilibria of the proposed network and give formulas for mapping binary patterns to these configurations. Most importantly, we demonstrate that the proposed 1D honeycomb Kuramoto Network exhibits exponential memory capacity and features no spurious memory states. Our future research includes analyzing the basin of attraction for the equilibria of the proposed network and designing algorithms  for complex computational tasks.

\bibliographystyle{unsrt}
\bibliography{alias,Main,New}

\appendix

For the proofs in this section it is useful to recall that the Kuramoto model \eqref{eq: Kuramoto model} with zero natural frequency is a gradient flow with respect to the energy
    \begin{align}\label{eq: kuramoto energy}
        f(\theta) = -\sum_{p,q=1}^{|\mathcal{V}|}  \frac{a_{ij}}{2} \cos(\theta_p - \theta_q).
    \end{align}
    In the cases we consider, where the weight matrix $A$ corresponds to an unweighted graph $\mathcal{G}=(\mathcal{V},\mathcal{E})$, we can rewrite the energy as
    \begin{align*}
        f(\theta) & =\sum_{p \in V} \sum_{q \in N_p} -\frac{1}{2}\cos(\theta_p - \theta_q)\\
        &= \sum_{(p,q) \in E} -\cos(\theta_p - \theta_q)
    \end{align*}
    where $N_p$ denotes the set of neighbors of $p$. The Kuramoto dynamics can then be rewritten as
    \begin{align*}
        \dot\theta = -\nabla f(\theta) = (\partial_{\theta_1} f(\theta), \partial_{\theta_2} f(\theta), ..., \partial_{\theta_n} f(\theta))^T
    \end{align*}
    where $\partial_{\theta_i} f$ denotes the partial derivative with respect to $\theta_i$.

\begin{lemma}\label{lem:phase cohesiveness near degree 2 nodes}
    Consider a Kuramoto model \eqref{eq: Kuramoto model} with identical natural frequencies and an unweighted network topology $\mc G = ( \mc V, \mc E)$. Any  stable phase-locked configuration $\theta^*$ of \eqref{eq: Kuramoto model} and any node $i \in V$ with exactly two neighbors $j, k$ must satisfy $\theta^*_i - \theta^*_{j} = \theta^*_{k} - \theta^*_i \in (-\frac{\pi}{2},\frac{\pi}{2})$.
\end{lemma}

\begin{proof}
    Looking at the system from a rotating frame, we can assume without loss of generality that $\omega_p=0$ for all $p \in \mc V$, in which case all phase-locked configurations become equilibria. We begin by solving for the equilibrium conditions on node $i$. Since $\theta^*$ is an equilibrium point, $\dot\theta^*_i=0$. Since the only neighbors of $i$ are $j$ and $k$, this equation reads as
    \begin{align}\label{eq:degree2-stationary-condition}
        \sin(\theta^*_{j} - \theta^*_{i})+\sin(\theta^*_{k} - \theta^*_i) = 0.
    \end{align}
    The two solutions to \eqref{eq:degree2-stationary-condition} are $\theta^*_j - \theta^*_{i} = -(\theta^*_{k} - \theta^*_i)$ or $\theta^*_{j} - \theta^*_{i} = \pi + \theta^*_{k} - \theta^*_{i}$, the latter of which is equivalent to $\theta^*_{j}=\pi + \theta^*_{k}$. At the end of the proof we show that if $\theta^\ast$ is a stable equilibrium, then $\theta^*_{j}\neq\pi + \theta^*_{k}$. For now, we assume that $\theta^*_{j} \neq \pi + \theta^*_{k}$, hence $\theta^*_j - \theta^*_{i} = -(\theta^*_{k} - \theta^*_i)$, and prove that $\theta^*_{k} - \theta^*_i \in (-\frac{\pi}{2},\frac{\pi}{2})$.
       
    Since the Kuramoto energy $f$ is a sum of cosine functions it is analytic, which implies via a result in \cite{absil2006gradientEQ} that $\theta^*$ is stable if and only if it is a local minimum of the energy function $f(\theta)$. Since $\theta^*$ is stable, it is a local minimum of $f$, which is true only if the Hessian of $f$ is positive definite. This implies
    \begin{align}
        0 \leq \partial^2_{\theta_i} f(\theta^*) &= \cos(\theta^*_{j} - \theta^*_{i})+\cos(\theta^*_{k} - \theta^*_i)\\
        &= 2\cos(\theta^*_{j} - \theta^*_{i})
    \end{align}
    hence we need $\cos(\theta^*_{j} - \theta^*_{i}) \geq 0$, which is true if and only if $\theta^*_{j} - \theta^*_{i} \in [-\frac{\pi}{2},\frac{\pi}{2}]$. Notice that $\theta^*_{j} - \theta^*_{i} \neq \pm \frac{\pi}{2}$, since otherwise we would have $(\theta^*_{j} - \theta^*_{i})-(\theta^*_{k} - \theta^*_i)=2(\theta^*_{j} - \theta^*_{i}) = \pi$, which would contradict the fact that $\theta^*_{j} \neq \pi + \theta^*_{k}$. Thus $\theta^*_{j} - \theta^*_{i} \in (-\frac{\pi}{2},\frac{\pi}{2})$, as desired.
    
    It remains to show that $\theta^*_{j} \neq \pi + \theta^*_{k}$. Suppose that $\theta^*_{j}=\pi + \theta^*_{k}$, and we will create a contradiction by showing that $\theta^{\ast}$ is not a local minimizer of $f(\theta)$. To show that $\theta^{\ast}$ is not a local minimizer, we will show for any arbitrary $\epsilon > 0$ that there is a point $\theta$ satisfying $\|\theta^\ast - \theta\|<\epsilon$ and $f(\theta) < f(\theta^*)$.

    We begin by constructing a point $\Tilde{\theta}^{\ast}$ which satisfies both $\|\Tilde{\theta}^{\ast} - \theta^{\ast}\|<\epsilon/2$ and $f(\Tilde{\theta}^{\ast}) = f(\theta^{\ast})$, but for which $\Tilde{\theta}^{\ast}$ is not an equilibrium point. To do so, we show for any state of the form $\Tilde{\theta}^* = (\theta^*_1, ..., \theta_{i-1}^*, \Tilde{\theta}_i, \theta_{i+1}^*, ..., \theta_{n}^*)$ that $f(\Tilde{\theta}^{\ast}) = f(\theta^\ast)$. Note that the difference $f(\Tilde{\theta}^{\ast}) - f(\theta^{\ast})$ only depends on the terms containing $\Tilde\theta^*_i$ and $\theta^*_i$, i.e
    \begin{align*}
        f(\Tilde{\theta}^{\ast}) - f(\theta^{\ast}) &= -(\cos(\theta_{k}^{\ast} - \Tilde{\theta}_{i}) + \cos(\theta_{j}^{\ast} - \Tilde{\theta}_{i})) \\ 
        & \;\;\;\;\;+(\cos(\theta_{k}^{\ast} - {\theta}^{\ast}_{i}) + \cos(\theta_{j}^{\ast} - {\theta}^{\ast}_{i})) 
    \end{align*}
    since we assumed $\theta_j^* = \theta_k^*+\pi$ we can rewrite the above as
    \begin{align*}
       f(\Tilde{\theta}^{\ast}) - f(\theta^{\ast}) &= -(\cos(\theta_{k}^{\ast} - \Tilde{\theta}_{i}) + \cos(\pi + \theta_{k}^{\ast} - \Tilde{\theta}_{i})) \\ 
        & \;\;\;\;\;+(\cos(\theta_{k}^{\ast} - {\theta}^{\ast}_{i}) + \cos(\pi + \theta_{k}^{\ast} - {\theta}^{\ast}_{i})) \\
        &= 0.
    \end{align*}
    % Next we construct a point $\Tilde{\theta}^{\ast}$ satisfying $\|\Tilde{\theta}^{\ast} - \theta^{\ast}\|<\epsilon/2$ and $f(\Tilde{\theta}^{\ast}) = f(\theta^{\ast})$, but $\Tilde{\theta}^{\ast}$ is not an equilibrium point. 
    % We start by perturbing $\theta_i$ to find a point near $\theta^*$ which has the same Kuramoto energy as $\theta^*$ but is not an equilibrium. The terms of $f$ that depend on $\theta_i$ are
    % \begin{align}\label{eq:terms of f that depend on $i$}
    %     -\cos(\theta_k-\theta_i)-\cos(\theta_j-\theta_i).
    % \end{align}
    % For any state of the form $\Tilde{\theta}^* = (\theta^*_1, ..., \theta_{i-1}^*, \Tilde{\theta}_i, \theta_{i+1}^*, ..., \theta_{n}^*)$ the terms of $f$ in \eqref{eq:terms of f that depend on $i$} become
    % \begin{align*}
    %     &-\cos(\theta_k^*-\Tilde\theta_i)-\cos(\theta_j^*-\Tilde\theta_i)
    % \end{align*}
    % and since we assumed that $\theta_j^* = \theta_k^*+\pi$ we can rewrite the above as
    % \begin{align}
    %     -\cos(\theta_k^*-\Tilde\theta_i&)-\cos((\theta_k^*+\pi)-\Tilde\theta_i)\\
    %     = &-\cos(\theta_k^*-\Tilde\theta_i)-\cos(\theta_k^*-\Tilde\theta_i + \pi)\\
    %     = &-\cos(\theta_k^*-\Tilde\theta_i)+\cos(\theta_k^*-\Tilde\theta_i)\\
    %     = &0.
    % \end{align}
    % This implies $f$ is constant with respect to $\theta_i$ when $\theta_j = \theta^*_j$ and $\theta_k=\theta^*_k$, which would imply that all point of the form $\Tilde{\theta}^* = (\theta^*_1, ..., \theta_{i-1}^*, \Tilde{\theta}_i, \theta_{i+1}^*, ..., \theta_{n}^*)$ satisfy
    % \begin{align} \label{eq: f\theta^*=f\Tilde\theta}
    %     f(\Tilde\theta^*)=f(\theta^*).
    % \end{align}
    We now find a choice of $\Tilde\theta^*$ which is not at equilibrium. At the state $\Tilde\theta^*$, the velocity of angle $j$ is
    \begin{align}
        \partial_{\theta_j}f(\Tilde\theta^*)\nonumber &= \sum_{p \in N_j} \sin(\Tilde\theta_p^* - \Tilde\theta_j^*)\nonumber\\ % label{eq: velocity intermediate step}
        &= \sin(\Tilde{\theta}_i - \theta_j^*) + \sum_{p \in N_j, p \neq i} \sin(\theta_p^* - \theta_j^*).\label{eq: velocity intermediate step}
    \end{align}
    Subtracting and adding $\sin(\theta_i^* - \theta_j^*)$ from \eqref{eq: velocity intermediate step} we obtain
    \begin{align}
        \partial_{\theta_j}f(\Tilde\theta^*) &= \sin(\Tilde{\theta}_i - \theta_j^*) - \sin(\theta_i^* - \theta_j^*) + \hspace*{-0.05in}\sum_{p \in N_j} \sin(\theta_p^* - \theta_j^*)\nonumber\\
        &= \sin(\Tilde{\theta}_i - \theta_j^*) - \sin(\theta_i^* - \theta_j^*)\label{eq: velocity of tilde angle j}
    \end{align}
where the last line holds because $\theta^*$ is an equilibrium so $\sum_{p \in N_j} \sin(\theta_p^* - \theta_j^*) = 0$. Since $\sin(\cdot)$ is nonconstant everywhere, for any value of $\theta_i^*$ there is a choice of $\Tilde{\theta}_i \in (\theta_i^*-\epsilon/2, \theta_i^*+\epsilon/2)$ such that
$$\partial_{\theta_j}f(\Tilde\theta^*) = \sin(\Tilde{\theta}_i - \theta_j^*) - \sin(\theta_i^* - \theta_j^*) \neq 0$$
i.e that $\Tilde\theta^*$ is not an equilibrium. Based on the definition of $\Tilde{\theta}^*$, it holds that $||\Tilde\theta^* - \theta^*||_2 = |\Tilde{\theta}_i - \theta_i| < \epsilon/2$.

Now let $\Tilde{\theta}^*$ be chosen such that it is not an equilibrium configuration and $||\Tilde\theta^* - \theta^*||_2 < \epsilon/2$. Since $\Tilde\theta^*$ is not an equilibrium it is not a local minimum. This means there is a point $\theta$ within an $\epsilon/2$-distance of $\Tilde\theta^*$ such that $f(\theta) < f(\Tilde\theta^*)$, and since $f(\Tilde\theta^*) = f(\theta^*)$ we know $f(\theta) < f(\theta^*)$. Since $||\Tilde{\theta^*}-\theta^*||_2 < \epsilon/2$ we know that
\begin{align*}
    ||\theta - \theta^*||_2 &\leq ||\theta - \Tilde\theta^*||_2 + ||\Tilde{\theta^*}-\theta^*||_2< \epsilon/2 + \epsilon/2= \epsilon.
\end{align*} Such a $\theta$ exists for any $\epsilon > 0$, hence $\theta^*$ is not a local minimum of $f$, giving us the desired contradiction. Thus it is not possible that $\theta^*_{j}=\pi + \theta^*_{k}$, which completes the proof.
\end{proof}

As a corollary in the special case when all edges contain a node of degree 2, we obtain the following proposition.

\begin{proposition}\label{prop: phase cohesiveness of our graphs}
    If every edge in an unweighted network $\mathcal{G}=(\mc V, \mc E)$ contains a node of degree 2, then any stable phase-locked configuration $\theta^*$ of \eqref{eq: Kuramoto model} with network topology $\mathcal{G}$ must satisfy $|\theta^*_i - \theta^*_j | < \pi/2$ for any adjacent nodes $i$ and $j$.
\end{proposition}

Proposition \ref{prop: phase cohesiveness of our graphs} only applies to certain subclasses of honeycomb networks such as those in Fig. \ref{fig: alternative network structure}. However, the Proposition doesn't apply to the networks in Fig. \ref{fig: network structure} since the junction nodes of degree $4$ are adjacent. For a result that applies to all tree-like honeycomb structures, we prove the following proposition.

\begin{proposition}\label{prop: phase cohesiveness of honeycombs}
    Let $n_c, m \in \N$ be given such that $n_c \geq 3$, $m \geq 1$ and let $\mathcal{G} = (\mathcal{V}, \mathcal{E})$ be an unweighted network constructed according to the following rule:
    \begin{enumerate}
        \item[(R)] $m$ cycles of size $n_c$, denoted $\mathcal{C}_p$ for $p=0, 1, ..., m-1$, are glued together in such a way that one node of $\mathcal{C}_p$ is identified with one node of $\mathcal{C}_{p+1}$ for each $p < m-1$.\label{rule: honeycomb construction}
    \end{enumerate}
    Any stable phase-locked configuration of the Kuramoto model \eqref{eq: Kuramoto model} with identical natural frequencies and network topology given by $\mathcal{G}$ must have $|\theta^*_i - \theta^*_j | < \pi/2$ for any adjacent nodes $i$ and~$j$.
\end{proposition}
\begin{proof}
    As in the proof of Lemma \ref{lem:phase cohesiveness near degree 2 nodes}, we consider a rotating frame such that $w_i = 0$ for all $i \in \mathcal{V}$, and all phase-locked configurations are equilibria. The proof is by induction on the number of cycles $m$. As a base case of $m=1$, $\mathcal{G}$ is a cycle of size at least 3. Then every node has degree exactly 2, so the result is obtained by Proposition \ref{prop: phase cohesiveness of our graphs}.
    
    For the inductive step, let $M$ be a fixed positive integer and suppose that all graphs constructed via Rule (R) with $m \leq M$ cycles only have phase-cohesive equilibria. Let $\mathcal{G}$ be a graph constructed via Rule (R) with $M+1$ cycles and we will show that any stable equilibrium $\theta^*$ is phase-cohesive. The strategy is to make use of the fact that the cycles are joined in a tree-like structure and the fact that the Kuramoto energy function is separable along each cycle, so that we can first prove the phase-cohesiveness of leaf-cycles, then remove the leaves and inductively treat the remaining subnetwork.

    Notice that the last cycle $\mathcal{C}_{m-1}=\mathcal{C}_M$ can only have one node $i_0$ connected to a node not in $\mathcal{C}_{m-1}$ since it is the last cycle to be glued (this one node will be connected to $\mathcal{C}_{m-2}$). Let $(i_0, i_1, ..., i_{n_c-1})$ denote the nodes of $\mathcal{C}_{m-1}$ in the cycle order such that $(i_l, i_{l+1}) \in \mathcal{E}$ for all $l=0,1,...,n_c-1$ and $(i_{n_c-1}, i_0) \in \mathcal{E}$. We know that $i_l$ for $l\neq 0$ must have degree 2, which implies by Lemma \ref{lem:phase cohesiveness near degree 2 nodes} the following:
    \begin{itemize}
        \item[(i)] the angle difference along any edge in $\mathcal{C}_{m-1}$ is less than $\pi/2$ and
        \item[(ii)] $\theta_{i_{l}}-\theta_{i_{l+1}} = \theta_{i_{l+1}}-\theta_{i_{l+2}}$ for each $l = 0,1,...,n_c-2$, where the index $l+2$ is taken modulo $n_c$.
    \end{itemize}
    Because of (i) it remains to show that the angle difference along any edge in $\mathcal{C}_{m-2}$, $\mathcal{C}_{m-3}$, ..., $\mathcal{C}_{0}$ is less than $\pi/2$. Because of (ii) we know that $\theta_{i_{0}}-\theta_{i_{1}} = \theta_{i_{n_c-1}}-\theta_{i_{0}}$.

    Recall from \cite{absil2006gradientEQ} that since $\theta^*$ is a stable equilibrium it is a local minimum of the Kuramoto energy $f$ described in \eqref{eq: kuramoto energy}. Notice that the energy $f$ is separable along the edges; that is, the energy $f$ can be decomposed as a sum of energies on each edge in $\mathcal{E}$. We separate $f$ as the sum of energies on edges in $\mathcal{C}_{p}$ for each $p \in \{0,1,...,m-1\}$, which we denote by $f_{\mathcal{C}_p}$ for each $p$, so that $f(\theta) = \sum_{p=0}^{m-1} f_{\mathcal{C}_p} (\theta)$. Given a fixed choice of $\theta_{i_0}^*, \theta_{i_1}^*, ..., \theta_{i_{n_c-1}}^*$, which we group together as $\theta^*_{\mathcal{C}_{m-1}}$, it is necessary that the remaining nodes, denoted $\theta^*_{-\mathcal{C}_{m-1}}$ are a local minimum of the sliced function $f(\theta^*_{\mathcal{C}_{m-1}}, \,\cdot\,)$, which takes as input $\theta^*_{-\mathcal{C}_{m-1}}$; otherwise we could perturb the coefficients $\theta^*_{-\mathcal{C}_{m-1}}$ in $\theta^*$ and reduce $f$, contradicting the fact that $\theta^*$ is a local minimum of $f$.
    
    We know that $f_{\mathcal{C}_p}$ only depends on edges between the nodes in $\mathcal{C}_{m-1}$, so it is constant with respect to $\theta^*_{-\mathcal{C}_{m-1}}$. This implies that minimizing $f(\theta^*_{\mathcal{C}_{m-1}}, \theta_{-\mathcal{C}_{m-1}})$ with respect to $\theta_{-\mathcal{C}_{m-1}}$ is equivalent to minimizing
    \begin{align*}
        f(\theta^*_{\mathcal{C}_{m-1}}, \theta_{-\mathcal{C}_{m-1}}) - f_{\mathcal{C}_p}(\theta^*_{\mathcal{C}_{m-1}}, \theta_{-\mathcal{C}_{m-1}})\\
        = \sum_{p=0}^{m-2} f_{\mathcal{C}_p} (\theta^*_{\mathcal{C}_{m-1}}, \theta_{-\mathcal{C}_{m-1}}).
    \end{align*}

    % Thus, any must have all edge-angle differences in $\mathcal{C_0}, ..., \mathcal{C}_{m-2}$ less than $\pi/2$. As stated previously, by (i) we only had left to prove that all edge-angle differences of $\theta^*$ in $\mathcal{C_0}, ..., \mathcal{C}_{m-2}$ are less than $\pi/2$, hence we are done.
    
    % %  must have all edge-angle differences in $\mathcal{C_0}, ..., \mathcal{C}_{m-2}$ less than $\pi/2$. As stated previously, by (i) we only had left to prove that all edge-angle differences of $\theta^*$ in $\mathcal{C_0}, ..., \mathcal{C}_{m-2}$ are less than $\pi/2$, hence we are done.

    Since the Kuramoto energy depends only on the relative phase of the angles, any local minimizer of the Kuramoto energy with one angle argument fixed is a local minimizer of the entire Kuramoto energy. Notice that $\theta^*_{-\mathcal{C}_{m-1}}$ is a local minimizer of $\sum_{p=0}^{m-2} f_{\mathcal{C}_p} (\theta^*_{\mathcal{C}_{m-1}}, \,\cdot\,)$, which is the energy of the network without cycle $\mathcal{C}_{m-1}$ when $\theta_{i_0}$ is fixed at $\theta^*_{i_0}$, which implies $(\theta^*_{i_0}, \theta^*_{-\mathcal{C}_{m-1}})$ is a local minimizer of the entire Kuramoto energy $\sum_{p=0}^{m-2} f_{\mathcal{C}_p}$. This implies that $(\theta^*_{i_0}, \theta^*_{-\mathcal{C}_{m-1}})$ is a stable equilibrium of the zero-frequency Kuramoto model with network topology given by joining $\mathcal{C}_0, ..., \mathcal{C}_{m-2}$. Per the previous paragraph, any angle differences $\theta^*_i - \theta^*_j$ in $\mathcal{C}_0, ..., \mathcal{C}_{m-2}$ must be less than $\pi/2$. This completes the proof.
    \end{proof}
\end{document}